%% file: root.tex
\documentclass[letterpaper, 10 pt, conference]{ieeeconf}  % Comment this line out
\IEEEoverridecommandlockouts                              % This command is only
\usepackage{graphics, xcolor} % for pdf, bitmapped graphics files
\usepackage{epsfig} % for postscript graphics files
\usepackage{mathptmx} % assumes new font selection scheme installed
\usepackage{times} % assumes new font selection scheme installed
\usepackage{amsmath} % assumes amsmath package installed
\usepackage{amssymb}  % assumes amsmath package installed
\usepackage{mathtools}
\usepackage{eso-pic}
\newtheorem{theorem}{Theorem}
\newtheorem{proposition}[theorem]{Proposition}
\newtheorem{remark}{Remark}
\title{\LARGE \bf
Schur-Neural KF: Learned Schur-Consistent Corrections to the Extended Kalman Filter
}

\author{Min Kim, Lianghao Cao, Soon-Jo Chung, Andrew M. Stuart%
\thanks{The authors are with the California Institute of Technology (Caltech), Pasadena, CA 91125 USA. e-mail: {\tt\small \{mink, lianghao, sjchung, astuart\}@caltech.edu}}%
}

\begin{document}

\AddToShipoutPictureFG*{%
\AtPageUpperLeft{%
\put(\LenToUnit{0.5in},\LenToUnit{-0.25in}){%
      \makebox(0,0)[lt]{%
        \parbox[t]{\dimexpr\paperwidth-1in\relax}{%
          \normalfont
          \fontsize{6.5}{7.5}\selectfont
          \color{black}
          \raggedright
          \copyright~2026 IEEE.  Personal use of this material is permitted.  Permission from IEEE must be obtained for all other uses, in any current or future media, including reprinting/republishing this material for advertising or promotional purposes, creating new collective works, for resale or redistribution to servers or lists, or reuse of any copyrighted component of this work in other works.\\
          Accepted to the 65th IEEE Conference on Decision and Control (CDC 2026).
          \par
        }%
      }%
    }%
  }%
}

\maketitle
\thispagestyle{empty}
\pagestyle{empty}

%%%%%%%%%%%%%%%%%%%%%%%%%%%%%%%%%%%%%%%%%%%%%%%%%%%%%%%%%%%%%%%%%%%%%%%%%%%%%%%%
\begin{abstract}
We present Schur-Neural KF (SN-KF), a learning-based correction to the extended Kalman filter (EKF) that preserves the probabilistic conditioning interpretation of the EKF. The method perturbs the predictive state--measurement cross-covariance and the Cholesky factor of the measurement noise covariance so that the resulting joint predictive covariance is always positive semidefinite. The positive semidefiniteness is ensured by a Schur complement-based parametrization. We instantiate the parametrization with a recurrent neural architecture whose matrix outputs are modulated by amplitude gates. We prove that incorporating a measurement does not increase the filter's state uncertainty, and show that no measurement can induce an arbitrarily large state correction relative to its statistical surprise. We also present a perturbative analysis suggesting SN-KF's structural strength in the data-scarce regime. We provide two numerical experiments to illustrate the practical benefits of SN-KF. In a two-radar experiment, enforcing Schur-consistency provides a much broader failure-free hyperparameter region and reduces RMSE for small training subsets, consistent with our theoretical analysis in the data-scarce regime. In the unicycle experiment, SN-KF achieves the best precision, recall, false alarm rate, and gated RMSE under innovation-based sensor-fault rejection.
\end{abstract}

%%%%%%%%%%%%%%%%%%%%%%%%%%%%%%%%%%%%%%%%%%%%%%%%%%%%%%%%%%%%%%%%%%%%%%%%%%%%%%%%
\section{INTRODUCTION}

State estimation for nonlinear stochastic systems is a core problem in control. When the underlying models are approximately known and the nonlinearities are moderate, the extended Kalman filter (EKF) remains one of the most widely used online estimators because of its computational simplicity. In practice, however, EKF performance can degrade substantially under model mismatch and linearization errors. This has motivated growing interest in learned Kalman-style estimators that use data to compensate for this performance degradation while retaining some of the structure of classical filtering.

Recent work has introduced learning or data-driven approaches into Kalman-style filtering in several ways. KalmanNet~\cite{ReSh2022} and KalmanFormer~\cite{ShCh2025} learn the Kalman gain directly from data using recurrent and transformer architectures, respectively. Split-KalmanNet~\cite{ChPa2023} computes the gain using the measurement Jacobian together with two recurrent networks that learn a priori state covariance and inverse innovation covariance. Cholesky-KalmanNet~\cite{KoSh2025} augments the RNNs of Split-KalmanNet with Cholesky-based output layers, thereby ensuring the positive definiteness of the learned matrices; triangular square-root factors for state covariance appear in the classical literature~\cite{Ca1973}. Recursive KalmanNet~\cite{MoFa2025} further augments this line of work by using two recurrent networks: one for the Kalman gain and the other for a noise covariance term, which are used for estimating the posterior state covariance. A-KIT~\cite{CoKl2025} uses a set-transformer to adapt the process noise covariance online within an EKF. Bayesian KalmanNet~\cite{DaRe2025} samples from a probabilistic ensemble of KalmanNet architectures, thereby directly estimating the posterior state covariance.
Learned components have also been incorporated into the ensemble Kalman filter~\cite{BaBa2026}. Finally, related work on learning state-space models includes deep Markov models~\cite{KrSh2017} and Neural EKF~\cite{LiLa2024}.

These developments show that learned Kalman-style estimators can be effective, but they also leave open a more structural question: whether learning can be introduced while preserving a valid conditioning interpretation of the Kalman update. This paper addresses that question through learned \emph{Schur-consistent corrections} to the EKF. Our architecture, called the Schur-Neural KF (SN-KF), introduces learning at the level of the predictive \emph{joint state--measurement covariance} rather than modifying the state-space model or directly learning the Kalman gain. A Schur complement construction guarantees that the corrected joint covariance remains positive semidefinite.

The contributions of this paper are threefold. \textit{(i)} We formulate learning-augmented EKF updates at the level of the predictive joint state--measurement covariance and derive a complete Schur complement parametrization that guarantees its positive semidefiniteness. We realize this parametrization by a recurrent architecture with learned amplitude gates; \textit{(ii)} we establish key theoretical properties of SN-KF: its posterior covariance is bounded above by the predictive covariance in the PSD order, and its state update cannot be arbitrarily large compared to the measurement's statistical surprise. A perturbation result further suggests that, under a nondegeneracy condition, our Schur complement construction exhibits reduced error in the data-scarce regime. Also, an explicit surrogate conditioning interpretation of the SN-KF and a finite-horizon closeness result to the nominal EKF are proved; \textit{(iii)} we evaluate SN-KF on two nonlinear benchmarks. In the two-radar task, it exhibits enhanced numerical stability by broadening the failure-free hyperparameter region, and achieves the lowest mean state-estimation error in the small-training-set experiment. In the unicycle task, it successfully calibrates the projected normalized innovation squared (NIS), and achieves the best precision, recall, false alarm rate, and gated RMSE under innovation-based sensor-fault rejection.

To explain the underlying idea, recall that the nominal EKF update can be interpreted as conditioning a joint Gaussian approximation of the stacked vector $[x_t^\top,z_t^\top]^\top$, where $x_t$ is the state and $z_t$ is the measurement. From this viewpoint, the update is determined by second-order predictive quantities, in particular the predictive covariance $P_{t|t-1}$, the predictive state--measurement cross-covariance $\bar C_{t|t-1}$, and the innovation covariance $\bar S_t$, where the overbar denotes the nominal EKF quantity before learned correction. When these quantities are modified by learned terms, it is natural to require them to still assemble into a positive semidefinite joint predictive covariance. This requirement is stronger than directly learning a correction gain in isolation, but it preserves the conditioning interpretation underlying the EKF update.

This structural constraint is also relevant for sensor-fault detection. Classical innovation-based fault detection often uses statistical properties of the Kalman filter innovations~\cite{MePe1971}; more broadly, model-based fault detection is often framed as residual generation followed by decision making based on the generated residuals~\cite{ChWi1984}. Accordingly, when learning is introduced into the update step, we seek to retain a well-defined innovation covariance so that downstream innovation-based monitoring and gating remain statistically coherent.

Motivated by these considerations, we propose learned \emph{Schur-consistent corrections} to the EKF. Rather than perturbing the nominal Kalman gain $\bar K_t = \bar C_{t|t-1}\bar S_t^{-1}$ directly, we perturb the nominal predictive state--measurement cross-covariance $\bar C_{t|t-1}$ and learn a measurement-side factor in a way that guarantees positive semidefiniteness of the predictive covariance of $(x_t,z_t)$ through a Schur complement construction. The corrected quantities define a joint predictive covariance~\eqref{eq:block_cov} which is positive semidefinite, so that the corrected innovation covariance and the corrected Kalman gain remain interpretable as Kalman quantities. The resulting method, which we refer to as \emph{Schur-Neural KF (SN-KF)}, keeps the EKF prediction step intact and learns correction terms only for the update step.

\noindent \textbf{Notation.}
Throughout this paper, an overbar denotes the \emph{nominal} EKF quantity before learned correction, while the corresponding unbarred symbol denotes the corrected quantity after learned perturbation. Vector norms are Euclidean norms, and matrix norms are spectral norms. $\mathbb{S}^n$ denotes the vector space of all $n$-by-$n$ real symmetric matrices, and $\mathbb{S}^n_{++}$ denotes its subset of all positive definite matrices.

\section{Schur-Neural Kalman Filter (SN-KF)}

SN-KF keeps the EKF prediction step unchanged and learns only the update step. At time $t$, a recurrent network outputs two corrections, $\Delta C_t$ and $\Delta L_t$, which induce a corrected \emph{joint} predictive covariance for $(x_t,z_t)$, from which the innovation covariance and Kalman gain are obtained. The joint covariance is ensured to be positive semidefinite by a Schur complement construction.

We first present the resulting update equations~\eqref{eq:schur_parametrization},~\eqref{eq:update_all} and then describe the recurrent architecture that parametrizes the corrections in Section~\ref{subsec:neural_architecture}.

\subsection{SN-KF Parametrization and Update Equations}
Consider the discrete-time nonlinear state-space model
\begin{align}
x_t &= f(x_{t-1},u_{t-1}) + w_{t-1}, \\
z_t &= h(x_t) + v_t,
\end{align}
with $f:\mathbb{R}^{n_x + n_u} \to \mathbb{R}^{n_x}$ and $h:\mathbb{R}^{n_x} \to \mathbb{R}^{n_z}$, where $x_t \in \mathbb{R}^{n_x}$ is the state, $z_t \in \mathbb{R}^{n_z}$ is the measurement, $u_t \in \mathbb{R}^{n_u}$ is a known input, and $w_t$ and $v_t$ are zero-mean Gaussian noises with covariances $Q_t \succeq 0$ and $R_t \succ 0$, respectively. Let $\hat x_{t|t-1}$ and $P_{t|t-1}$ denote the EKF predictive mean and covariance,
and let
\begin{align}
F_{t-1} &\triangleq \left.\frac{\partial f}{\partial x}\right|_{\hat x_{t-1|t-1},u_{t-1}}, \\
H_t &\triangleq \left.\frac{\partial h}{\partial x}\right|_{\hat x_{t|t-1}},
\end{align}
where we assumed $f(\cdot, u_{t-1})$ and $h$ are differentiable. The EKF prediction step, which is left unchanged in SN-KF, is
\begin{align}
\hat x_{t|t-1} &= f(\hat x_{t-1|t-1},u_{t-1}), \\
P_{t|t-1} &= F_{t-1} P_{t-1|t-1} F_{t-1}^\top + Q_{t-1}.
\end{align}

We define the joint predictive covariance of the stacked pair $(x_t,z_t)$ as
\begin{equation}
\Sigma_{t|t-1} \triangleq
\begin{bmatrix}
P_{t|t-1} & C_{t|t-1} \\
C_{t|t-1}^\top & S_t
\end{bmatrix},
\label{eq:block_cov}
\end{equation}
where $C_{t|t-1}$ and $S_t$ are the corrected cross-covariance and innovation covariance. Our notion of \emph{Schur-consistency} is precisely that
\begin{equation}
\Sigma_{t|t-1} \succeq 0.
\end{equation}
Assuming $P_{t|t-1} \succ 0$, the Schur complement characterization implies that \eqref{eq:block_cov} is positive semidefinite if and only if
\begin{equation}
S_t - C_{t|t-1}^\top P_{t|t-1}^{-1} C_{t|t-1} \succeq 0.
\label{eq:schur_condition}
\end{equation}

This motivates the following learned parametrization for SN-KF, in which $\Delta C_t$ and $\Delta L_t$ are the learned components:
\begin{subequations}
\begin{align}
C_{t|t-1} &= \bar C_{t|t-1} + \Delta C_t, \label{eq:C_pert}\\
L_t &= \bar L_t + \Delta L_t,
\label{eq:L_pert}\\
S_t &= C_{t|t-1}^\top P_{t|t-1}^{-1} C_{t|t-1} + L_t L_t^\top. \label{eq:S_pert}
\end{align}
\label{eq:schur_parametrization}
\end{subequations}
Here, $\bar L_t$ is the Cholesky factor of $R_t \succ 0$, $\Delta L_t$ is lower-triangular, and the nominal EKF predictive covariances are given by
\begin{subequations}
\begin{align}
\bar C_{t|t-1} &= P_{t|t-1} H_t^\top, \\
\bar S_t &= H_t P_{t|t-1} H_t^\top + R_t=
\bar C_{t|t-1}^\top P_{t|t-1}^{-1} \bar C_{t|t-1} + R_t. 
\end{align}
\end{subequations}
Because $L_t L_t^\top \succeq 0$, the Schur complement in \eqref{eq:schur_condition} is positive semidefinite. Hence, \eqref{eq:block_cov} is automatically Schur-consistent. Moreover, the parametrization is \emph{complete} in the sense that it could parametrize any $(C_{t|t-1}, S_t)$ for which the assembled $\Sigma_{t|t-1}$ is positive semidefinite.

The SN-KF update step is
\begin{subequations}
\begin{align}
\hat x_{t|t} &= \hat x_{t|t-1} + K_t \nu_t, \label{eq:state_update}\\
\begin{split}
P_{t|t} &= P_{t|t-1} - K_t C_{t|t-1}^\top - C_{t|t-1} K_t^\top + K_t S_t K_t^\top\\
&=  P_{t|t-1} - C_{t|t-1}S_t^{-1} C_{t|t-1}^\top\\
&= \left(I-K_t C_{t|t-1}^\top P_{t|t-1}^{-1} \right)P_{t|t-1}\left(I-K_t C_{t|t-1}^\top P_{t|t-1}^{-1} \right)^\top \\
&\qquad +K_t L_t L_t^\top K_t^\top,
\label{eq:cov_update}
\end{split}
\end{align}
\label{eq:update_all}
\end{subequations}
where the innovation is defined as
\begin{equation}
\nu_t \triangleq z_t - h(\hat x_{t|t-1})
\end{equation}
and the Kalman gain associated with the learned joint covariance is defined as
\begin{equation}
K_t \triangleq C_{t|t-1} S_t^{-1}.
\end{equation}
In implementation, we make $L_t$ non-singular by clamping so that $S_t \succ 0$ and $K_t$ is well-defined.

\noindent \textbf{A learned gain correction baseline.} 

For comparison, we define a learned gain correction baseline that predicts an additive correction to the nominal EKF gain:
\begin{align}
\bar K_t &= \bar C_{t|t-1} \bar S_t^{-1}, \\
K_t &= \bar K_t + \Delta K_t,
\end{align}
followed by the state update \eqref{eq:state_update} and the standard Joseph form covariance update equation:
\begin{equation}
P_{t|t} = (I-K_tH_t)P_{t|t-1}(I-K_tH_t)^\top + K_t R_t K_t^\top.
\end{equation}

Unlike~\eqref{eq:C_pert}--\eqref{eq:S_pert}, the learned gain correction baseline does not explicitly enforce that the learned correction arises from a positive semidefinite joint predictive covariance of $(x_t,z_t)$, so a valid conditioning interpretation is not built into the parametrization.

\subsection{Neural Architecture}
\label{subsec:neural_architecture}
The proposed architecture keeps the EKF prediction step unchanged and uses a recurrent neural network only to parametrize the perturbations $\Delta C_t$ and $\Delta L_t$. At time $t$, the encoder receives a history vector
\begin{equation}
\psi_t \triangleq
\begin{bmatrix}
\nu_{t-1} \\
\hat z_{t-1|t-2} \\
u_{t-1}
\end{bmatrix}\text{, }\text{ } \hat z_{t-1|t-2} \triangleq h(\hat x_{t-1|t-2}),
\label{eq:encoder_history_vec}
\end{equation}
and updates a latent feature through a gated recurrent unit (GRU)~\cite{ChVa2014},
\begin{equation}
r_t = \operatorname{GRU}(\psi_t,r_{t-1}).
\end{equation}
This recurrent state summarizes recent innovations, predicted measurements, and controls. For $t=0$, the previous-step quantities are initialized as $r_{-1}=0$, $\nu_{-1}=0$, $\hat z_{-1|-2}=0$, and $u_{-1}=0$.

\paragraph*{SN-KF}
The Schur-consistent model uses two neural heads driven by the shared recurrent feature $r_t$. Each head uses separate multilayer perceptrons for its matrix output and gate vector, with spectral normalization applied to each MLP's last linear layer.

First, the \emph{cross-covariance head} outputs a raw matrix $M_t^{(C)} \in \mathbb{R}^{n_x \times n_z}$ together with a gate vector $g_t^{(C)} \in \mathbb{R}^{n_z}$. We define
\begin{equation}
\Delta C_t = M_t^{(C)} \operatorname{diag}\!\left(\alpha_C \,\sigma(g_t^{(C)})\right),
\label{eq:deltaC_net}
\end{equation}
where $\sigma(\cdot)$ is the standard logistic function and $\alpha_C>0$ is a trainable correction scale. The diagonal gate controls the per-measurement amplitude.

Second, the \emph{innovation-factor head} outputs a vector of $n_z(n_z+1)/2$ entries together with a gate vector $g_t^{(L)} \in \mathbb{R}^{n_z}$; the former is reshaped into a lower-triangular matrix $M_t^{(L)} \in \mathbb{R}^{n_z\times n_z}$. We use the same idea:
\begin{equation}
\Delta L_t = M_t^{(L)} \operatorname{diag}\!\left(\alpha_L \,\sigma(g_t^{(L)})\right),
\label{eq:deltaL_net}
\end{equation}
with a trainable positive scale $\alpha_L$. The effective factor is then
\begin{equation}
L_t = \bar L_t + \Delta L_t.
\end{equation}
In implementation, we additionally clamp the diagonal entries of the lower-triangular matrix $L_t$ from below by a small positive constant (specifically, we use $10^{-4}$). This alters the final $\Delta L_t$ value and makes the final $L_t$ non-singular, which implies $S_t \succ 0$; hence $K_t = C_{t|t-1} S_t^{-1}$ is well-defined.

The amplitude gates in~\eqref{eq:deltaC_net}~and~\eqref{eq:deltaL_net} partially separate the directions of the corrections from their column-wise amplitudes.

\paragraph*{Learned gain correction baseline}
To compare against direct residual learning of the Kalman gain, we use the same recurrent encoder but replace the Schur-consistent heads by a single \emph{gain head}. This head outputs a raw matrix $M_t^{(K)} \in \mathbb{R}^{n_x\times n_z}$ and a gate vector $g_t^{(K)} \in \mathbb{R}^{n_z}$, and defines
\begin{align}
\Delta K_t&=M_t^{(K)}\operatorname{diag}\!\left(\alpha_K \,\sigma(g_t^{(K)})\right),\\
K_t &= \bar K_t + \Delta K_t,
\end{align}
where $\alpha_K>0$ is again a trainable positive scale. Each of the global correction scales, $\alpha_C$, $\alpha_L$, and $\alpha_K$ is parametrized as the softplus of a scalar trainable parameter and is frozen after offline training.

\section{Theoretical Properties of SN-KF}
We begin by showing that the learned corrections preserve positive semidefiniteness of the \emph{joint} predictive covariance.
\begin{proposition}
Assume $P_{t|t-1} \succ 0$. For any matrices $\Delta C_t$ and $\Delta L_t$, define $C_{t|t-1}$ and $S_t$ by~\eqref{eq:C_pert}--\eqref{eq:S_pert}. Then $\Sigma_{t|t-1} \succeq 0$. Consequently, the updated covariance
\begin{equation}
P_{t|t} = P_{t|t-1} - C_{t|t-1} S_t^{-1} C_{t|t-1}^\top
\label{eq:conditional_cov}
\end{equation}
is also positive semidefinite whenever $S_t \succ 0$. Finally, the parametrization~\eqref{eq:schur_parametrization} is rich enough to parametrize any $(C_{t|t-1}, S_t)$ for which the assembled $\Sigma_{t|t-1}$ is positive semidefinite.
\end{proposition}
\begin{proof}
Using \eqref{eq:S_pert},
\begin{equation}
S_t - C_{t|t-1}^\top P_{t|t-1}^{-1} C_{t|t-1}=L_t L_t^\top \succeq 0.
\end{equation}
Hence the Schur complement of the upper-left block in \eqref{eq:block_cov} is positive semidefinite, which implies $\Sigma_{t|t-1} \succeq 0$. The conditional covariance formula \eqref{eq:conditional_cov} is the Schur complement of $S_t$. The final claim follows from the fact that any real $A \succeq 0$ has a decomposition of the form $A=L L^\top$ for some $L$ real and lower-triangular.
\end{proof}

The next result analyzes the SN-KF update equations to show that these updates remain reasonable. Concretely, the posterior mean update in normalized coordinates is bounded by a suitably normalized innovation vector; moreover, the posterior covariance remains no greater than the predictive covariance.

\begin{theorem}
Assume $P_{t|t-1}\succ0$ and let $L_t\in\mathbb R^{n_z\times n_z}$ be non-singular. For every innovation vector $\nu_t\in\mathbb R^{n_z}$,
\begin{equation}
\|P_{t|t-1}^{-1/2}K_t\nu_t\| \le \min\left(\frac12\,\|L_t^{-1}\nu_t\|, \sqrt{\nu_t^\top S_t^{-1} \nu_t} \right).
\end{equation}
Also,
\begin{equation}
0\prec P_{t|t}\preceq P_{t|t-1}.
\label{eq:posterior_cov_order}
\end{equation}
\end{theorem}

\begin{proof}
Write $P\triangleq P_{t|t-1}$, $C\triangleq C_{t|t-1}$, $L\triangleq L_t$, $S \triangleq S_t$, $K \triangleq K_t$, and define $M\triangleq P^{-1/2} C L^{-T}$. We have $C = P^{1/2} M L^\top$, and hence $S = L (I + M^\top M) L^{\top}$. Therefore, we have
\begin{equation*}
P^{-1/2}K L=P^{-1/2}CS^{-1}L=M(I+M^\top M)^{-1}.
\end{equation*}
Now let $M=U\Sigma V^\top$ be a singular value decomposition of $M$. Then
\begin{equation*}
M(I+M^\top M)^{-1}=U\,\Sigma(I+\Sigma^\top \Sigma)^{-1}V^\top.
\end{equation*}

Therefore, each singular value of $P^{-1/2}K L$ is at most $1/2$, because $x/(1+x^2) \le 1/2$ for all $x \ge 0$. Combined with $K^\top P^{-1}K = S^{-1}-S^{-1}L L^\top S^{-1}\preceq S^{-1}$, this proves the first assertion. For the second assertion, note that $P_{t|t} = P^{1/2} \left[ I - M(I+M^\top M)^{-1} M^\top  \right]P^{1/2} = P^{1/2} (I+MM^\top)^{-1} P^{1/2}$, where the second equality comes from the Woodbury matrix identity. Pre-multiplying and post-multiplying the symmetric matrix $P^{1/2}$ to $0 \prec (I+MM^\top)^{-1}\preceq I$ gives~\eqref{eq:posterior_cov_order}.
\end{proof}

\begin{remark}[Comparison to direct gain-learning] In the scalar case $n_x=n_z=1$ with $P_{t|t-1}=1$, $H_t=1$, $R_t=1$, suppose a direct gain-learning architecture outputs an arbitrary scalar gain $k_t$. The Joseph form covariance update becomes $P_{t|t}=1-2k_t+2k_t^2$.
Therefore, $P_{t|t}>P_{t|t-1}$ whenever $k_t<0$ or $k_t>1$, and $P_{t|t}\to\infty$ as $|k_t|\to\infty$.
\end{remark}

We now give a heuristic explanation for why SN-KF may be beneficial in the data-scarce regime. Fix a time $t$, and denote $P=P_{t|t-1}\succ0$, $L=L_t$, where $L$ is non-singular. Let $C^\dagger \in \mathbb{R}^{n_x \times n_z}$ be a reference (``ground-truth'') cross-covariance, and define $S^\dagger\triangleq C^{\dagger\top}P^{-1}C^\dagger+LL^\top$, $K^\dagger\triangleq C^\dagger(S^\dagger)^{-1}$.
For a fixed perturbation direction $H\in\mathbb R^{n_x\times n_z}$ and an $\varepsilon>0$, let $C_\varepsilon \triangleq C^\dagger+\varepsilon H$.

We compare the SN-KF gain $K^{\mathrm{SN}} \triangleq C_\varepsilon \left(C_\varepsilon^\top P^{-1}C_\varepsilon+LL^\top\right)^{-1}$ with $K^{\mathrm{TD}}\triangleq C_\varepsilon(S^\dagger)^{-1}$, calculated using the true denominator $S^\dagger$. For any $K\in\mathbb R^{n_x\times n_z}$, define the normalized \emph{gain error}
\begin{equation*}
\mathcal E(K) \triangleq \left\|P^{-1/2}(K-K^\dagger)(S^\dagger)^{1/2}\right\|_F^2.
\end{equation*}

\begin{proposition}
We have $\mathcal E\!\left(K^{\mathrm{SN}}\right)= \mathcal E\!\left(K^{\mathrm{TD}}\right)-\varepsilon^2\Delta(H)+ o(\varepsilon^2)$ as $\varepsilon\to0+$, where $\Delta(H)\ge0$. If $C^{\dagger\top}P^{-1}H + H^\top P^{-1}C^\dagger \neq0$, we have $\Delta(H)>0$.
\end{proposition}
\begin{remark}[Connection to the data-scarce regime]
Suppose that the learned cross-covariance error has the typical scale $C_{t|t-1}-C^\dagger \sim N^{-1/2}$, where $N$ is the training set size. Heuristically, identifying $\varepsilon$ with $N^{-1/2}$ gives a reduction in the gain error of order $N^{-1}$, which could be significant in the small-$N$ regime.
\end{remark}

\begin{proof}
Set $W\triangleq(S^\dagger)^{-1}$ and $\Gamma\triangleq C^{\dagger\top}P^{-1}H+H^\top P^{-1}C^\dagger$.
Define the SN-KF gain map $\mathcal K(C)\triangleq C(C^\top P^{-1}C+LL^\top)^{-1}$. From the product rule, we have the following directional derivative:
\begin{equation*}
D\triangleq D\mathcal K({C^\dagger})[H] = HW-C^\dagger W\Gamma W.
\end{equation*}
Hence,
$K^{\mathrm{SN}} = K^\dagger +\varepsilon D +o(\varepsilon)$,
whereas
$K^{\mathrm{TD}}=K^\dagger+\varepsilon HW$. Therefore,
\begin{align*}
\mathcal E\!\left(K^{\mathrm{SN}}\right)&=\varepsilon^2 \left\| P^{-1/2}D (S^\dagger)^{1/2} \right\|_F^2 +o(\varepsilon^2),\\
\mathcal E\!\left(K^{\mathrm{TD}}\right)&=\varepsilon^2 \left\|P^{-1/2}H(S^\dagger)^{-1/2}\right\|_F^2,
\end{align*}
where the first equality comes from Cauchy--Schwarz.

Define $\Delta(H) \triangleq \left\| L^\top W \Gamma W^{1/2} \right\|_F^2 \ge0$. Since $\Gamma=\Gamma^\top$ and $C^{\dagger\top}P^{-1}C^\dagger=S^\dagger-LL^\top$, one can check that
\begin{equation*}
\begin{split}
\left\|P^{-1/2}D(S^\dagger)^{1/2} \right\|_F^2&=\operatorname{tr}(P^{-1} HWH^\top) - \operatorname{tr}(LL^\top W\Gamma W \Gamma W)\\
&=\left\|P^{-1/2}H(S^\dagger)^{-1/2}\right\|_F^2 - \Delta (H).
\end{split}
\end{equation*}
This completes the proof.
\end{proof}

Our fourth result gives a probabilistic interpretation of the SN-KF parametrization. The SN-KF update arises as conditioning in an explicit surrogate Gaussian model.

\begin{proposition}[Surrogate Gaussian interpretation]
Fix $t$, and condition on all information available before assimilating $z_t$. Assume $P_{t|t-1}\succ0$ and that $L_t$ is non-singular. Define auxiliary random vectors $(\widetilde x_t,\widetilde z_t)$ by
\begin{align}
\widetilde x_t&\sim\mathcal N(\hat x_{t|t-1},P_{t|t-1}),\\
\widetilde z_t&=h(\hat x_{t|t-1})+C_{t|t-1}^\top P_{t|t-1}^{-1}
\bigl(\widetilde x_t-\hat x_{t|t-1}\bigr)+\eta_t,\\
\eta_t&\sim \mathcal N(0,L_tL_t^\top),
\end{align}
with $\eta_t$ independent of $\widetilde x_t$. Then the following hold.
\begin{enumerate}
\item $(\widetilde x_t,\widetilde z_t)$ is jointly Gaussian with covariance $\Sigma_{t|t-1}$ from \eqref{eq:block_cov}.

\item For every $z\in\mathbb R^{n_z}$,
\begin{equation}
\widetilde x_t \mid (\widetilde z_t=z)\sim\mathcal N\!\Bigl(\hat x_{t|t-1}+K_t\bigl(z-h(\hat x_{t|t-1})\bigr),P_{t|t} \Bigr).
\label{eq:surrogate_conditional}
\end{equation}
In particular, at $z=z_t$, the conditional mean is \eqref{eq:state_update}, and the conditional covariance is \eqref{eq:cov_update}.

\item If $p_{\widetilde z_t}$ denotes the probability density function of $\widetilde z_t$, then, evaluated at the actual measurement $z_t$,
\begin{equation}
-\log p_{\widetilde z_t}(z_t)=\frac12 \nu_t^\top S_t^{-1}\nu_t+\frac12 \log\det S_t + \frac{n_z}{2}\log(2\pi).
\end{equation}

\item Let
\begin{equation*}
q_t^-=\mathcal N(\hat x_{t|t-1},P_{t|t-1}),
\end{equation*}
and let $q_t^+(\cdot\mid z)$ be the Gaussian law in \eqref{eq:surrogate_conditional}. Then
\begin{equation}
\begin{split}
\mathbb E\!\left[D_{\mathrm{KL}}\!\left(q_t^+(\cdot\mid \widetilde z_t)\,\|\,q_t^-\right)\right]
&= \frac 1 2 \log\frac{\det P_{t|t-1}}{\det P_{t|t}}
\\
&=\frac12 \log\frac{\det S_t}{\det(L_tL_t^\top)},
\label{eq:surrogate_kl}
\end{split}
\end{equation}
where the expectation is with respect to $\widetilde z_t$.
\end{enumerate}
\end{proposition}
\begin{proof}
Write $P \triangleq P_{t|t-1}$, $C \triangleq C_{t|t-1}$, $L \triangleq L_t$, $S \triangleq S_t$, and $\Sigma \triangleq \Sigma_{t|t-1}$. First, note that \begin{equation*}\begin{bmatrix} \widetilde x_t \\ \widetilde z_t \end{bmatrix} = \begin{bmatrix}I & 0 \\ C^\top P^{-1} & I\end{bmatrix}\begin{bmatrix} \widetilde x_t \\ \eta_t \end{bmatrix} + \begin{bmatrix} 0 \\ h(\hat x_{t|t-1})-C^\top P^{-1} \hat x_{t|t-1} \end{bmatrix},\end{equation*}
so $(\widetilde x_t, \widetilde z_t)$, being an affine image of a jointly Gaussian vector, is jointly Gaussian; its covariance is given by \begin{equation*}\begin{bmatrix} I & 0 \\C^\top P^{-1} & I \end{bmatrix} \begin{bmatrix} P & 0 \\ 0 & L L^{\top} \end{bmatrix} \begin{bmatrix}I & P^{-1} C \\ 0 & I\end{bmatrix} = \Sigma.\end{equation*}
This proves item 1. For item 2, note that $S \succ 0$, and use the Gaussian conditioning formula~\cite[Proposition~3.13]{Ea2007}. Item 3 immediately follows from noting that the marginal law of $\widetilde z_t$ is $\mathcal N(h(\hat x_{t|t-1}),S)$.

For item 4, first note that $P\succ0$ and $S-C^\top P^{-1}C = LL^\top \succ0$ imply $\Sigma \succ 0$. In particular, the Schur complement of $S$ in $\Sigma$, i.e., $P-CS^{-1}C^\top = P_{t|t}$ is also positive definite. Applying the Gaussian KL divergence formula~\cite{SaSt2023} gives
\begin{align*}
&D_{\mathrm{KL}}\left(q_t^+(\cdot\mid \widetilde z_t)\,\|\,q_t^- \right)\\
&= \frac 1 2 \left[\operatorname{tr}\left(P^{-1}(P-CS^{-1}C^\top)\right)-n_x
+\log\frac{\det P}{\det(P-CS^{-1}C^\top)} \right]\\
&\quad +\frac12\left(\widetilde z_t-h(\hat x_{t|t-1})\right)^\top S^{-1}C^\top P^{-1}CS^{-1} \left(\widetilde z_t-h(\hat x_{t|t-1})\right).
\end{align*}

Now we take expectations with respect to $\widetilde z_t$. Recall that
\begin{equation*}
\mathbb E\left[(\widetilde z_t-h(\hat x_{t|t-1}))(\widetilde z_t-h(\hat x_{t|t-1}))^\top\right]
=S,
\end{equation*}
and hence
\begin{align*}
&\mathbb E \left[\left(\widetilde z_t-h(\hat x_{t|t-1})\right)^\top S^{-1}C^\top P^{-1}CS^{-1} \left(\widetilde z_t-h(\hat x_{t|t-1})\right)\right ]\\
&=\mathbb E \left[\operatorname{tr} 
\left(S^{-1}C^\top P^{-1}CS^{-1}\left(\widetilde z_t-h(\hat x_{t|t-1})\right) \left(\widetilde z_t-h(\hat x_{t|t-1})\right)^\top \right)
\right] \\
&= \operatorname{tr}\left( S^{-1} C^\top P^{-1} C\right).
\end{align*}

On the other hand,
\begin{equation*}
\operatorname{tr}\left(P^{-1}(P-CS^{-1}C^\top)\right)-n_x=-\operatorname{tr}\left(P^{-1}CS^{-1}C^\top\right).
\end{equation*}
By the cyclic invariance of trace, the expected KL divergence simplifies to
\begin{align*}
\mathbb E\left[ D_{\mathrm{KL}}\left(q_t^+(\cdot\mid \widetilde z_t)\,\|\,q_t^-\right)\right]
&=\frac 1 2 \log\frac{\det P}{\det(P-CS^{-1}C^\top)}\\&=\frac12 \log\frac{\det P_{t|t-1}}{\det P_{t|t}}.
\end{align*}

It remains to prove the second equality in \eqref{eq:surrogate_kl}. Since $P$ and $S$ are invertible, the Schur complement determinant formula gives $\det\Sigma = \det(P)\det\!\bigl(S-C^\top P^{-1}C\bigr) = \det(P)\det(LL^\top)$ and also $\det\Sigma = \det(S)\det\!\bigl(P-CS^{-1}C^\top\bigr)$. Comparing these two expressions for $\det\Sigma$, we obtain
\begin{equation*}
\det(P)\det(LL^\top) = \det(S)\det(P-CS^{-1}C^\top).
\end{equation*}
Rearranging the terms completes the proof.
\end{proof}

When $\Delta C_t = 0$ and $\Delta L_t = 0$ for all $t$, SN-KF reduces to the nominal EKF. Our final theorem states that, over a fixed finite horizon, small correction terms lead to trajectories that remain close to the nominal EKF trajectory.

\begin{theorem}[Finite-horizon closeness to nominal EKF]
Fix $T \ge 0$, a measurement sequence $\{z_t\}_{t=0}^{T}$, an input sequence $\{u_t\}_{t=0}^{T-1}$, and a common Gaussian prior with mean $\hat x_{0|-1}$ and covariance $P_{0|-1} \succ 0$. Let $\{(\bar x_{t|t},\bar P_{t|t})\}_{t=0}^{T}$ denote the nominal EKF trajectory, and let $\{(\hat x_{t|t},P_{t|t})\}_{t=0}^{T}$ denote the corresponding SN-KF trajectory with corrections $\Delta C_t$ and $\Delta L_t$.

Assume that:
\begin{enumerate}
\item for each $t=0,\dots,T-1$, the map $x \mapsto f(x,u_t)$ is $C^2$, and $h$ is $C^2$;
\item there exist constants $q,r>0$ such that $Q_t \succeq qI$ for $t=0,\dots,T-1$ and $R_t \succeq rI$ for $t=0,\dots,T$.
\end{enumerate}

Define
\begin{equation*}
\varepsilon_C \triangleq \max_{0\le t\le T}\|\Delta C_t\|, \qquad \varepsilon_L \triangleq \max_{0\le t\le T}\|\Delta L_t\|.
\end{equation*}
Then there exist constants $\varepsilon_0>0$ and $c_T>0$, independent of the particular correction sequence $\Delta C_t$ and $\Delta L_t$, such that whenever $\varepsilon_C+\varepsilon_L \le \varepsilon_0$, the SN-KF recursion is well-defined up to time $t=T$ and satisfies
\begin{equation}
\max_{0\le t\le T} \Bigl(\|\hat x_{t|t}-\bar x_{t|t}\|+\|P_{t|t}-\bar P_{t|t}\|\Bigr)
\le c_T(\varepsilon_C+\varepsilon_L).
\label{eq:finite_horizon_closeness}
\end{equation}
\end{theorem}

\begin{proof}
Define the norm $\|(x,P)\|_* \triangleq \|x\|+\|P\|$ for $(x,P) \in \mathbb{R}^{n_x} \times \mathbb{S}^{n_x}$. Denote the nominal EKF predicted and posterior pairs by
\begin{equation*}
\bar X_t^- \triangleq (\bar x_{t|t-1},\bar P_{t|t-1}), \qquad \bar X_t^+ \triangleq (\bar x_{t|t},\bar P_{t|t}), \qquad t=0,\dots,T.
\end{equation*}

We first note that all nominal EKF covariances are positive definite. First, $P_{0|-1}\succ0$ by assumption. For each $t \ge 0$, if $\bar P_{t|t-1} \succ 0$ then $\bar S_t  \succeq rI \succ 0$ so the EKF update is well-defined.
Moreover, by the Woodbury matrix identity, $\bar P_{t|t} = \left(\bar P_{t|t-1}^{-1}+\bar H_t^\top R_t^{-1} \bar H_t\right)^{-1}$, where $\bar H_t$ denotes the Jacobian $\frac{\partial h}{\partial x}$ evaluated at the EKF's predicted state at time $t$. Hence, $\bar P_{t|t} \succ 0$. Also, for each $t \ge 0$, $\bar P _{t|t} \succ 0$ implies $\bar P_{t+1|t}  \succeq qI \succ 0$. By induction, all nominal predicted and posterior covariances are positive definite. Therefore, there exists $\rho>0$ such that each closed $\rho$-ball in $\mathbb{R}^{n_x} \times \mathbb{S}^{n_x}$ centered at $\bar X_t ^\pm$ is contained in the open set $\mathbb{R}^{n_x} \times \mathbb{S}^{n_x}_{++}$. Denote the union of all these balls as $N_\rho$, and note that this set is compact.

For each $t=0,\dots,T-1$, define the prediction map (which is used for both the nominal EKF and SN-KF)
\begin{equation*}
\Phi_t(x,P) \triangleq \Bigl(f(x,u_t),\,F(x,u_t)\,P\,F(x,u_t)^\top + Q_t \Bigr),
\end{equation*}
where $F(x,u_t)\triangleq \frac{\partial f}{\partial x}(x,u_t)$.
Since $x\mapsto f(x,u_t)$ is $C^2$, the map $\Phi_t$ is $C^1$. For each $t=0,\dots,T$, define the SN-KF update map
\begin{equation*}
U_t(x,P,\Delta C,\Delta L)\triangleq (x^+,P^+)
\end{equation*}
where
\begin{align*}
H(x) \triangleq \frac{\partial h}{\partial x}(x),& \quad \bar C(x,P) \triangleq P H(x)^\top, \\
C = \bar C(x,P) + \Delta C,& \quad L = \bar L_t + \Delta L, \\
S = C^\top P^{-1} C + LL^\top,& \quad K = C S^{-1}, \\
x^+ = x + K\bigl(z_t-h(x)\bigr),& \quad P^+ = P - K C^\top - C K^\top + K S K^\top.
\end{align*}

The natural domain of $U_t$ is characterized by the inequalities $P \succ 0$ and $S \succ 0$, hence is an open set; note that $U_t$ is $C^1$ on this set.

Now, define the closed set
\begin{equation*}
B \triangleq \left\{ (\Delta C,\Delta L): \|\Delta C\|+\|\Delta L\| \le \sqrt{r}/2 \right\}.
\end{equation*}
Suppose $(\Delta C, \Delta L) \in B$. Then, for each unit-norm vector $x$, we have $\| \bar L_t ^\top x \|^2 = x^\top R_t x \ge r$,  so that $\| L^\top x \| \ge \| \bar L_t ^\top x \| - \| \Delta L^\top x\| \ge \sqrt r - \| \Delta L \|  \ge \sqrt{r}/2$. Therefore, $x^\top LL^\top x = \| L^\top x\|^2  \ge \frac{r}{4}$ for all unit-norm $x$. Hence, we have $LL^\top \succeq (r/4) I$ and $S = C^\top P^{-1} C + LL^\top \succeq \frac r4 I \succ 0$; therefore the domain of $U_t$ contains the compact set $N_\rho \times B$. 

As $C^1$ maps are Lipschitz on compact sets, there exist constants $L_\Phi,L_U>0$ such that for all relevant $t$,
\begin{align}
\|\Phi_t(X)-\Phi_t(Y)\|_* &\le L_\Phi \|X-Y\|_*,
\label{eq:pred_lipschitz}
\\
\|U_t(X,\Delta)-U_t(Y,0)\|_* &\le L_U \left( \|X-Y\|_* +  \|\Delta C\|+\|\Delta L\| \right),
\label{eq:update_lipschitz}
\end{align}
for all $X,Y\in N_\rho$ and all $\Delta=(\Delta C,\Delta L)\in B$.

Now let $X_t^- \triangleq (\hat x_{t|t-1},P_{t|t-1})$ and $X_t^+ \triangleq (\hat x_{t|t},P_{t|t})$ for $t=0,\dots,T$ denote the SN-KF predicted and posterior pairs, respectively. Define $d_t^- \triangleq \| X_t^- - \bar X_t^- \|_*$, $d_t^+ \triangleq \| X_t^+ - \bar X_t^+  \|_*$. Recursively define the following constants $a_i$:
\begin{equation*}
a_{-1} \triangleq 0, \qquad a_{t} \triangleq L_U (L_\Phi\, a_{t-1} + 1), \qquad t=0,\dots,T,
\end{equation*}
and set
\begin{equation*}
c_T \triangleq \max \Bigl( \{a_t:\ 0\le t\le T\} \cup \{L_\Phi a_t:\ 0\le t\le T-1\} \Bigr).
\end{equation*}
Finally, choose $\varepsilon_0  \triangleq \min\left\{\frac{\sqrt{r}}{2},\frac{\rho}{c_T}\right\} >0$.

Now, we prove by induction that if $\varepsilon \triangleq \varepsilon_C + \varepsilon_L \le\varepsilon_0$, then for all $t=0,\dots,T$,

\begin{equation}
d_t^- \le L_\Phi a_{t-1}\varepsilon, \qquad X_t^- \in N_\rho,
\label{eq:induction_minus}
\end{equation}
and 
\begin{equation}
d_t^+ \le a_t \varepsilon, \qquad X_t^+ \in N_\rho.
\label{eq:induction_plus}
\end{equation}

At $t=0$, we have $X_0^-=\bar X_0^-$, so $d_0^-=0$ and
$X_0^- \in N_\rho$. Also, $\|\Delta C_0\|+\|\Delta L_0\| \le \varepsilon \le \frac{\sqrt{r}}{2}$, so $\Delta_0 \triangleq(\Delta C_0,\Delta L_0)\in B$. Hence, by \eqref{eq:update_lipschitz},
\begin{equation*}
d_0^+ = \|U_0(X_0^-,\Delta_0)-U_0(\bar X_0^-,0)\|_* \le L_U \varepsilon = a_0 \varepsilon.
\end{equation*}
Since $\varepsilon\le \rho/c_T\le \rho/a_0$, this also implies $X_0^+\in N_\rho$. Therefore, both \eqref{eq:induction_minus} and \eqref{eq:induction_plus} hold for $t=0$.

Now, assume both \eqref{eq:induction_minus} and \eqref{eq:induction_plus} hold for some $t\in\{0,\dots,T-1\}$. Because $X_t^+,\bar X_t^+\in N_\rho$, we may apply \eqref{eq:pred_lipschitz} to obtain
\begin{equation*}
d_{t+1}^- = \|\Phi_t(X_t^+)-\Phi_t(\bar X_t^+)\|_* \le L_\Phi d_t^+ \le L_\Phi a_t\varepsilon.
\end{equation*}
Since $\varepsilon\le \rho/c_T\le \rho/(L_\Phi a_t)$, it follows that $X_{t+1}^- \in N_\rho$. This proves \eqref{eq:induction_minus} for $t+1$.

Next, we have $\|\Delta C_{t+1}\|+\|\Delta L_{t+1}\| \le \sqrt{r}/2$, so $\Delta_{t+1}\triangleq(\Delta C_{t+1},\Delta L_{t+1})\in B$. Because both $X_{t+1}^-$ and $\bar X_{t+1}^-$ lie in $N_\rho$, we may apply \eqref{eq:update_lipschitz} to get $d_{t+1}^+ = \|U_{t+1}(X_{t+1}^-,\Delta_{t+1})-U_{t+1}(\bar X_{t+1}^-,0)\|_* \le L_U (d_{t+1}^- + \varepsilon) \le L_U (L_\Phi a_t \varepsilon + \varepsilon) = a_{t+1}\varepsilon$.
Since $\varepsilon\le \rho/c_T \le \rho/a_{t+1}$, we also obtain $X_{t+1}^+ \in N_\rho$. Therefore, \eqref{eq:induction_plus} holds for $t+1$, completing the induction.

Therefore, for all $t=0,\dots,T$, we have $d_t^+ \le a_t(\varepsilon_C + \varepsilon_L) \le c_T(\varepsilon_C + \varepsilon_L)$,
and we conclude that the claim~\eqref{eq:finite_horizon_closeness} holds.
\end{proof}

\input{experimental_results}

\section{CONCLUSIONS}
We presented SN-KF, a structured perturbation to the update step of the extended Kalman filter. By learning corrections to the predictive joint covariance of the state--measurement pair, the method retains a coherent Gaussian-conditioning interpretation together with an explicit innovation covariance. Theoretical analysis shows that SN-KF has many desirable properties: incorporating a measurement does not increase the filter's state uncertainty, and no measurement can induce an arbitrarily large state correction relative to its statistical surprise. Also, a perturbative argument suggests that SN-KF exhibits reduced error in the data-scarce regime. Finally, we prove a surrogate Gaussian model interpretation of SN-KF and also a closeness theorem relative to the nominal EKF.

The two experiments expose practical benefits of SN-KF. In the two-radar study, SN-KF provides a much broader failure-free tuning region and reduced RMSE for data-scarce training sets. In the unicycle study, SN-KF attains the best precision, recall, false alarm rate, and gated RMSE under the theoretical $\chi^2_6$ NIS gate. These results suggest that SN-KF makes learned Kalman updates both more reliable to train and more trustworthy for innovation-based decisions.

%%%%%%%%%%%%%%%%%%%%%%%%%%%%%%%%%%%%%%%%%%%%%%%%%%%%%%%%%%%%%%%%%%%%%%%%%%%%%%%%
%\section{ACKNOWLEDGMENTS}

%%%%%%%%%%%%%%%%%%%%%%%%%%%%%%%%%%%%%%%%%%%%%%%%%%%%%%%%%%%%%%%%%%%%%%%%%%%%%%%%

\vspace{-5pt}
\bibliographystyle{IEEEtran}
\bibliography{references}

\end{document}

%% file: experimental_results.tex
\section{Experimental Results}
We next evaluate two practical consequences of Schur-consistency: training robustness and reliability of the resulting innovation statistics. The two-radar experiment examines numerical stability across correction scale initializations and state estimation accuracy in the data-scarce regime. The unicycle experiment asks whether a Schur-consistent learned joint covariance improves NIS calibration and downstream fault-gated estimation.

\subsection{Experimental Setup}
We compare SN-KF with a capacity-matched \emph{no-Schur} ablation, learned gain correction, and an inflation-tuned EKF. SN-KF and no-Schur use the same neural architecture (hence the same number of parameters) and neural network initialization method; no-Schur omits the coupling of the corrected cross-covariance into the innovation covariance:
\begin{equation}
S_t^{\rm NS}=H_t P_{t|t-1} H_t^\top +L_tL_t^\top,\qquad K_t^{\rm NS}=C_{t|t-1}(S_t^{\rm NS})^{-1}.
\end{equation}
SN-KF uses the last equality in \eqref{eq:cov_update}, and no-Schur uses the Joseph form covariance update $P_{t|t}^{\rm NS} = (I-K_t^{\rm NS}H_t)P_{t|t-1}(I-K_t^{\rm NS}H_t)^\top + K_t^{\rm NS}L_tL_t^\top(K_t^{\rm NS})^\top$.
All neural methods use full backpropagation through the filtering recursion, gradient clipping at one, and then an AdamW update; they are initialized so that $\Delta C_t$, $\Delta L_t$, $\Delta K_t$ are all zero, and are trained and selected using state MSE. Within each benchmark, they share the peak learning rate, learning rate schedule, weight decay, and mini-batch size. The first epoch warms the learning rate from 1\% to 100\% of its peak, followed by cosine decay to 1\%. SN-KF and no-Schur have the same number of neural parameters; compared to these, the learned gain correction baseline has a parameter count differing by less than 0.2\%.

For the inflation-tuned EKF, we rescale the predictive covariance by $P_{t|t-1}\leftarrow\gamma^2P_{t|t-1}$ before computing the innovation covariance and Kalman gain, with $\gamma>0$ selected using validation state MSE. In the two-radar experiment, this rescaling is applied only when measurements are available.
\subsection{Two Radars: Correction Scale Tuning and Training Stability}
We consider a 5-dimensional constant-turn-rate vehicle model observed by two fixed range--bearing radars with intermittent sensing. Measurements are jointly available from both radars with probability 25\% at each of the 40 filtering steps. The state is $x_t=[p_{x,t},p_{y,t},v_{x,t},v_{y,t},\omega_t]^\top$, and the radar locations are $(-20,-5)$ and $(20,-5)$. The time step for discretization is $\Delta t = 0.35$.
The dynamics have no control input. In this benchmark, all neural methods replace $u_{t-1}$ in \eqref{eq:encoder_history_vec} with the current measurement-availability flag $m_t\in\{0,1\}$. The GRU is updated at every step; when measurements are unavailable, the innovation input to the GRU at the next time step is set to zero.
The initial state is Gaussian with mean $[0,0,1,0,0.05]^\top$ and covariance $P_0=\operatorname{diag}(0.3^2,0.3^2,0.15^2,0.15^2,0.02^2)$. We use noise covariances $Q=\operatorname{diag}(0.3^2,0.03^2,0.5^2,0.05^2,0.002^2)$ and $R=\operatorname{diag}(0.5^2,0.0175^2,0.1^2,0.0175^2)$, with measurement order $[r_1,\theta_1,r_2,\theta_2]$. For $t\ge1$, the pairs $(w_{t-1},v_t)$ are i.i.d. Gaussian with cross-covariance $N=Q^{1/2}BR^{1/2}$, where $B\in\mathbb R^{5\times4}$ has entries $B_{11}=-0.49$, $B_{13}=B_{31}=0.85$, and $B_{33}=0.49$, with all other entries being zero. All filters use the true dynamics and measurement functions, the true initial distribution, $Q$, and $R$, but assume $N=0$.

We first use 700 training and 175 validation trajectories to select the initial global correction scale hyperparameters. We report results on 1,000 test trajectories. The validation and test sets are common throughout the two-radar experiment.
SN-KF and no-Schur are evaluated on the same initial 36-point grid $\alpha_C\in\{0.001,\allowbreak 0.01,\allowbreak 0.1,\allowbreak 1,\allowbreak 10,\allowbreak 100\}$ and $\alpha_L\in\{0.001, 0.01,\allowbreak 0.1,\allowbreak 1,\allowbreak 10,\allowbreak 100\}$ using three random seeds, hence $36 \times 3 = 108$ training runs per method. For the no-Schur ablation, 42 out of 108 training runs numerically failed (``Initial Succ." in Table~\ref{tab:radar_hyperparameter_robustness}). Learned gain correction is evaluated on 11 $\alpha_K$ candidates in $\{0.01,\allowbreak 0.03,\allowbreak 0.1,\allowbreak 0.3,\allowbreak 1,\allowbreak 3,\allowbreak 10,\allowbreak 30,\allowbreak 100,\allowbreak 300,\allowbreak 1000\}$, and hence $11 \times 3 = 33$ training runs are performed. All neural methods use the same peak learning rate $5 \times 10^{-3}$.
After this 15-epoch grid search, we further refine the chosen hyperparameters using 30 epochs per run, selecting $(\alpha_C,\alpha_L)=(0.316228,1)$ for both SN-KF and no-Schur and $\alpha_K=0.547723$ for learned gain correction. Refinement comprises 93 runs over 31 configurations for each of SN-KF/no-Schur and 18 runs over 6 gain correction configurations. SN-KF and no-Schur share all evaluated configurations throughout the refinement.

We then evaluate hyperparameter robustness with respect to a fixed, small training set (last two columns in Table~\ref{tab:radar_hyperparameter_robustness}). Around each selected setting, we evaluate nine nearby configurations on a fixed 30-trajectory subset of the original 700 training trajectories using three fresh random seeds: a $3\times3$ grid for SN-KF and no-Schur, and nine nearby $\alpha_K$ values for learned gain correction. All nine local configurations remain valid for SN-KF. No-Schur has no valid local configurations, while gain correction has two.

\begin{table}[t]
\vspace{5pt}
\caption{Numerical stability of training for the two-radar benchmark. ``Initial Succ.'' counts successful runs over the initial grids before refinement. For neural methods, ``Val. RMSE'' and ``Test RMSE'' are three-seed means at the final selected scales, trained using 30 training epochs. A local configuration is valid only if all three 30-epoch runs complete successfully. The EKF uses validation-selected inflation $\gamma =0.905$.}
\label{tab:radar_hyperparameter_robustness}
\centering
\setlength{\tabcolsep}{1.55pt}
\scriptsize
\begin{tabular}{@{}lrrrrr@{}}
\hline
Method & Initial Succ. & Val. RMSE & Test RMSE & Local Succ. & Valid \\
\hline
SN-KF & \textbf{108/108} & 1.5393 & 1.7006 & \textbf{27/27} & \textbf{9/9} \\
No-Schur & 66/108 & \textbf{1.5359} & \textbf{1.6913} & 7/27 & 0/9 \\
Gain correction & 18/33 & 1.6222 & 1.7708 & 19/27 & 2/9 \\
EKF (infl.-tuned) & -- & 1.7993 & 1.9045 & -- & -- \\
\hline
\end{tabular}
\vspace{-10pt}
\end{table}

\begin{table}[t]
\vspace{5pt}
\caption{Results from training on 100 different small subsets. Parentheses in ``Test RMSE'' give the sample standard deviation. The EKF is fixed across subsets, so the corresponding standard deviation is zero.}
\label{tab:radar_subset_training}
\centering
\setlength{\tabcolsep}{1.45pt}
\scriptsize
\begin{tabular}{@{}lrrr@{}}
\hline
Method & Success & Test RMSE & Worst RMSE\\
\hline
SN-KF & \textbf{100/100} & \textbf{1.873(30)} & 1.952 \\
No-Schur & 48/100 & 1.908(43) & 2.017 \\
Gain correction & 61/100 & 1.909(20) & 2.023 \\
EKF (infl.-tuned) & -- & 1.904(0) & \textbf{1.904} \\
\hline
\end{tabular}
\vspace{-15pt}
\end{table}

Finally, to test state estimation performance in the data-scarce regime, we form one hundred 30-trajectory training subsets from the original training set and train each method once per subset for 30 epochs (Table~\ref{tab:radar_subset_training}). For each neural method, we report the mean test RMSE over the successful runs. ``Worst RMSE'' is the largest test RMSE among such runs. In conclusion, Schur-consistency makes training substantially more reliable, while SN-KF remains competitive in terms of the state estimation performance throughout this experiment, achieving the lowest test RMSE in the small-training-set experiment.

\subsection{Unicycle: NIS Calibration and Fault-Gated Estimation}
We consider a 4-dimensional unicycle with state vector $x_t=[p_{x,t},p_{y,t},\theta_t,v_{f,t}]^\top$, where $v_{f,t}$ denotes forward velocity. Each measurement consists of three range--bearing pairs from three fixed beacons, hence is six-dimensional. The ground-truth system and all filters share the same dynamics and measurement models. The only model mismatch hidden from the filters is a nonzero process--measurement noise cross-covariance $N$. We generate the process and measurement noises as $w_{t-1}=G\xi_t+\epsilon_t$ and $v_t=D\xi_t+\eta_t$, where $\xi_t\sim\mathcal N(0,I_2)$, $\epsilon_t \sim \mathcal N(0,Q-GG^\top)$, and $\eta_t \sim \mathcal N(0,R-DD^\top)$ are mutually independent. Note that $\operatorname{Cov}(w_{t-1})=Q$, $\operatorname{Cov}(v_t)=R$, and $N\triangleq \operatorname{Cov}(w_{t-1},v_t)=GD^\top \neq 0$.

We use 128/128 train/validation trajectories of length 50. All neural methods use a peak learning rate of $3 \times 10^{-4}$. The validation set was used to choose the initial values of the global correction scales by minimizing the mean best validation MSE over three random seeds, with 15 epochs per tuning run. The selected initial values are $(\alpha_C,\alpha_L)=(0.166667,1.83333)$ for SN-KF, $(\alpha_C,\alpha_L)=(0.433333,1.16667)$ for no-Schur, and $\alpha_K=3.22222$ for gain correction; the selected EKF inflation is $\gamma=1.02$. Then, each neural method is trained from scratch for 30 epochs using five fresh random seeds, with its global correction scales initialized to the previously selected values.

\paragraph{NIS Calibration}
\begin{table}[t]
\vspace{5pt}
\caption{Projected NIS over all 50 steps of healthy trajectories. The $\chi^2_2$ reference values are $2/0.95/0.99$.}
\label{tab:unicycle_nis}
\centering
\setlength{\tabcolsep}{3.1pt}
\scriptsize
\begin{tabular}{@{}lrrr@{}}
\hline
Method & Proj. NIS & C95 & C99 \\
\hline
SN-KF & \textbf{1.945} & \textbf{0.9539} & \textbf{0.9912} \\
No-Schur & 1.330 & 0.9866 & 0.9984 \\
Gain correction & 1.419 & 0.9832 & 0.9975 \\
EKF (infl.-tuned) & 1.493 & 0.9808 & 0.9972 \\
\hline
\end{tabular}
\vspace{-5pt}
\end{table}

We report $\operatorname{NIS}^{\rm proj}_t=(U^\top\nu_t)^\top (U^\top S_tU)^{-1} (U^\top\nu_t)$, whose reference distribution is $\chi^2_2$. Here, $U\in\mathbb R^{6\times2}$ consists of orthonormal columns spanning the two-dimensional column space of $D$. On a 1500-trajectory test set, SN-KF attains NIS 1.945 with empirical 95\%/99\% coverages $0.9539/0.9912$, closely matching the references $2/0.95/0.99$ without any further recalibration (Table~\ref{tab:unicycle_nis}).

\paragraph{Fault-Gated Estimation}
We next test the filters under common bearing faults. At a faulty time step, we add the same angular offset to all three bearing measurements, while leaving the three range measurements unchanged. At each of the remaining 45 time steps after the first 5 steps, a fault occurs with probability $0.1$. After accepting the first five updates, each filter computes the NIS using the six measurement components, and accepts the measurement update only when $\operatorname{NIS}_t$ is no greater than the 99th percentile of $\chi^2_6$. If the NIS exceeds this threshold, the filter skips that measurement update. For testing, we use 1500 test trajectories, which are the corrupted-measurement versions of the test set used for the projected NIS experiment.

\begin{table}[t]
\caption{Innovation-based gating for $[9,11]\sigma$ common bearing faults, where $\sigma=0.03$ rad; FAR = False Alarm Rate. Gated RMSE is computed over the full 50-step trajectory.}
\label{tab:unicycle_fault_gating}
\vspace{-5pt}
\centering
\setlength{\tabcolsep}{1.35pt}
\scriptsize
\begin{tabular}{@{}lrrrr@{}}
\hline
Method & Precision & Recall & FAR & Gated RMSE \\
\hline
SN-KF & \textbf{0.808} & \textbf{0.937} & \textbf{0.0249} & \textbf{0.2104} \\
No-Schur & 0.534 & 0.692 & 0.0677 & 0.2523 \\
Gain correction & 0.523 & 0.814 & 0.0829 & 0.2800 \\
EKF (infl.-tuned) & 0.713 & 0.839 & 0.0377 & 0.2229 \\
\hline
\end{tabular}
\vspace{-15pt}
\end{table}

SN-KF has the highest precision and recall, together with the lowest FAR and gated RMSE (Table~\ref{tab:unicycle_fault_gating}). Relative to no-Schur, it improves precision, recall, and FAR by 27.44, 24.49, and 4.28 percentage points, respectively, and reduces gated RMSE by 16.63\%. The unicycle experiments suggest that SN-KF successfully maintains a coherent statistical meaning of the Kalman innovation vector, which is useful for sensor-fault detection and downstream fault-gated state estimation.

%% file: references.bib
@article{ReSh2022,
title = "{KalmanNet}: Neural Network Aided {Kalman} Filtering for Partially Known Dynamics",
author = "Guy Revach and Nir Shlezinger and Xiaoyong Ni and Escoriza, \{A. L.\} and van Sloun, \{Ruud J.G.\} and Eldar, \{Yonina C.\}",
year = "2022",
doi = "10.1109/TSP.2022.3158588",
language = "English",
volume = "70",
pages = "1532--1547",
journal = "IEEE Trans. Signal Process.",
issn = "1053-587X",
}

@ARTICLE{ChPa2023,
  author={Choi, Geon and Park, Jeonghun and Shlezinger, Nir and Eldar, Yonina C. and Lee, Namyoon},
  journal={IEEE Trans. Veh. Technol.}, 
  title={{Split-KalmanNet}: A Robust Model-Based Deep Learning Approach for State Estimation}, 
  year={2023},
  volume={72},
  number={9},
  pages={12326--12331},
  doi={10.1109/TVT.2023.3270353}}

@inproceedings{MoFa2025,
    author = {Hassan Mortada and Cyril Falcon and Yanis Kahil and Mathéo Clavaud and Jean-Philippe Michel},
    title = {Recursive {KalmanNet}: Deep Learning-Augmented {Kalman} Filtering for State Estimation with Consistent Uncertainty Quantification},
    booktitle = {33rd European Signal Processing Conference},
    year = {2025},
    pages = {885--889}
}

@inproceedings{ChVa2014,
    title = "Learning Phrase Representations using {RNN} Encoder{--}Decoder for Statistical Machine Translation",
    author = {Cho, Kyunghyun  and
      van Merri{\"e}nboer, Bart  and
      Gulcehre, Caglar  and
      Bahdanau, Dzmitry  and
      Bougares, Fethi  and
      Schwenk, Holger  and
      Bengio, Yoshua},
    booktitle = "Proceedings of the 2014 Conference on Empirical Methods in Natural Language Processing ({EMNLP})",
    year = "2014",
    address = "Doha, Qatar",
    doi = "10.3115/v1/D14-1179",
    pages = "1724--1734"
}

@article{Ca1973,
author = {Carlson, Neal A.},
title = {Fast triangular formulation of the square root filter},
journal = {AIAA Journal},
volume = {11},
number = {9},
pages = {1259--1265},
year = {1973},
doi = {10.2514/3.6907},
eprint = { 
        https://doi.org/10.2514/3.6907
}
}

@ARTICLE{KoSh2025,
  author={Ko, Minhyeok and Shafieezadeh, Abdollah},
  journal={IEEE Signal Process. Lett.}, 
  title={{Cholesky-KalmanNet}: Model-Based Deep Learning With Positive Definite Error Covariance Structure}, 
  year={2025},
  volume={32},
  number={},
  pages={326--330},
  doi={10.1109/LSP.2024.3519265}}

@ARTICLE{DaRe2025,
  author={Dahan, Yehonatan and Revach, Guy and Dunik, Jindrich and Shlezinger, Nir},
  journal={IEEE Trans. Signal Process.}, 
  title={Bayesian {KalmanNet}: Quantifying Uncertainty in Deep Learning Augmented {Kalman} Filter}, 
  year={2025},
  volume={73},
  number={},
  pages={2558--2573},
  doi={10.1109/TSP.2025.3581703}}

@ARTICLE{ChWi1984,
  author={Chow, E. and Willsky, A.},
  journal={IEEE Trans. Autom. Control}, 
  title={Analytical redundancy and the design of robust failure detection systems}, 
  year={1984},
  volume={29},
  number={7},
  pages={603--614},
  doi={10.1109/TAC.1984.1103593}}

@article{MePe1971,
title = {An innovations approach to fault detection and diagnosis in dynamic systems},
journal = {Automatica},
volume = {7},
number = {5},
pages = {637--640},
year = {1971},
issn = {0005-1098},
doi = {https://doi.org/10.1016/0005-1098(71)90028-8},
author = {R.K. Mehra and J. Peschon}
}

@article{BaBa2026,
title = {Learning enhanced ensemble filters},
journal = {Journal of Computational Physics},
volume = {547},
pages = {114550},
year = {2026},
issn = {0021-9991},
doi = {https://doi.org/10.1016/j.jcp.2025.114550},
author = {Eviatar Bach and Ricardo Baptista and Edoardo Calvello and Bohan Chen and Andrew Stuart}
}

@inproceedings{KrSh2017,
  title={Structured Inference Networks for Nonlinear State Space Models},
  author={Krishnan, Rahul G and Shalit, Uri and Sontag, David},
  booktitle={AAAI},
  year={2017}
}

@book{Ea2007,
  author    = {Morris L. Eaton},
  title     = {Multivariate Statistics: A Vector Space Approach},
  year      = {2007},
  publisher = {Institute of Mathematical Statistics},
}

@ARTICLE{ShCh2025,
    
AUTHOR={Shen, Siyuan  and Chen, Jichen  and Yu, Guanfeng  and Zhai, Zhengjun  and Han, Pujie },
           
TITLE={{KalmanFormer}: using transformer to model the {Kalman} Gain in {Kalman} Filters},
          
JOURNAL={Frontiers in Neurorobotics},
          
VOLUME={18},
  
YEAR={2025},
  
DOI={10.3389/fnbot.2024.1460255},
  
ISSN={1662-5218}}

@book{SaSt2023, place={Cambridge}, series={London Mathematical Society Student Texts}, title={Inverse Problems and Data Assimilation}, publisher={Cambridge University Press}, author={Sanz-Alonso, Daniel and Stuart, Andrew and Taeb, Armeen}, year={2023}, collection={London Mathematical Society Student Texts}}

@article{LiLa2024,
author = {Wei Liu and Zhilu Lai and Kiran Bacsa and Eleni Chatzi},
title ={Neural extended {Kalman} filters for learning and predicting dynamics of structural systems},

journal = {Structural Health Monitoring},
volume = {23},
number = {2},
pages = {1037--1052},
year = {2024},
doi = {10.1177/14759217231179912}
}

@article{CoKl2025,
   title={Adaptive {Kalman}-Informed Transformer},
   volume={146},
   ISSN={0952-1976},
   DOI={10.1016/j.engappai.2025.110221},
   journal={Engineering Applications of Artificial Intelligence},
   publisher={Elsevier BV},
   author={Cohen, Nadav and Klein, Itzik},
   year={2025},
   month=apr, pages={110221} }
